\documentclass[9pt,twocolumn,twoside]{IEEEtran} 
\IEEEoverridecommandlockouts                        
\usepackage{times,color,graphicx, setspace, bbm, mathdots,mathrsfs, amssymb,latexsym,amsfonts,amsmath,stmaryrd, caption,cite,setspace,amsmath,url,pgf,tikz}
\usepackage[dvips]{psfrag}

\usepackage{textcomp}
\usetikzlibrary{shapes,arrows}
\tikzset{
 	block/.style = {draw, thick, rectangle, minimum height = 3em,
	minimum width = 3em},
	sum/.style = {draw, circle, node distance = 2cm}, 
	input/.style = {coordinate}, 
 	output/.style = {coordinate} 
}
\newtheorem{theorem}{Theorem}
\newtheorem{proof}{Proof}

\newtheorem{example}{Example}
\newtheorem{assumption}{Assumption}

\newcommand{\R}{{\mathbb{R}}}

\newcommand{\ie}{{i.e., }}

\newcommand{\T}{{\text{T}}}

\ifCLASSINFOpdf
\else
\fi

\begin{document}

\title{Fundamental Limits and Tradeoffs  in Autocatalytic Pathways}
\author{Milad Siami$^1$, Nader Motee$^2$, Gentian Buzi$^3$, Bassam Bamieh$^4$,  Mustafa Khammash$^5$ and John C. Doyle$^6$ 

\thanks{$^{1}$ M. Siami is with the Institute for Data, Systems, and Society, Massachusetts Institute of Technology, Cambridge, MA 02319. Email:  {\tt\small siami@mit.edu}.}
\thanks{$^2$ N. Motee is with the Department of Mechanical Engineering and Mechanics, Packard Laboratory, Lehigh University, Bethlehem, PA 18015. Email addresses:  {\tt\small motee@lehigh.edu} }
\thanks{$^3$ G. Buzi is with State University of New York at Fredonia. Email addresses:  {\tt\small Gentian.Buzi@fredonia.edu}}
\thanks{$^4$ B.~Bamieh is with the Department of Mechanical Engineering, University of California, Santa Barbara, CA 93106, USA.  Email address: {\tt\small bamieh@engineering.ucsb.edu.}}
\thanks{$^5$ Mustafa Khammash is with the Department of Biosystems Science and Engineering, ETH Zurich, Mattenstrasse 26, 4058, Basel, Switzerland. Email address: {\tt\small mustafa.khammash@bsse.ethz.ch}}
\thanks{$^6$ J. C.~Doyle is with the Control and Dynamical Systems, California Institute of Technology, Pasadena, CA, 91125, USA. Email address: {\tt\small doyle@cds.caltech.edu}}
}
\maketitle
\thispagestyle{empty}
\pagestyle{empty}

\begin{abstract}
This paper develops some basic principles to study autocatalytic networks and exploit their structural properties in order to characterize their inherent  fundamental limits and tradeoffs. In a dynamical system with autocatalytic structure, the system's output is necessary to catalyze its own production. We consider a simplified model of Glycolysis pathway as our motivating application. First, the properties of these class of pathways are investigated through a simplified two-state model, which is obtained by lumping all the intermediate reactions into a single intermediate reaction. Then, we  generalize our results to autocatalytic pathways that are composed of a chain of enzymatically catalyzed intermediate reactions. We explicitly derive a hard limit on the minimum achievable $\mathcal L_2$-gain disturbance attenuation and a hard limit on its minimum required output energy. Finally, we show how these resulting hard limits lead to some fundamental tradeoffs between transient and steady-state behavior of the network and its net production.

\end{abstract}

\section{Introduction}
\allowdisplaybreaks

Dynamic autocatalysis mechanisms are inherent to several real-world dynamical networks including most of the planet's cells from bacteria to human, engineered networks as well as economic systems \cite{chandra11, Autocatalytic, buzi11, buzi2010control}. In an interconnected control system with autocatalytic structure, the system's product (output) is necessary to power and catalyze its own production. The destabilizing effects of such ``positive'' autocatalytic feedback can be countered by negative regulatory feedback. There have been some recent interest to study models of glycolysis pathway as  an example of an autocatalytic dynamical network in biology that generates adenosine triphospate (ATP), which is the cell's energy currency and is consumed by different mechanisms in the cell \cite{chandra11,motee6}. Other examples of autocatalytic networks include engineered power grids whose machinery are maintained using their own energy product as well as financial systems which operate based on generating monetary profits by investing money in the market. Recent results show that there can be severe theoretical hard limits on the resulting performance and robustness in autocatalytic dynamical networks. It is shown that the consequence of such tradeoffs stems from the autocatalytic structure of the system \cite{chandra11,BuziD10,motee6}. 

	The recent interest in understanding fundamental limitations of feedback in complex interconnected dynamical networks from biological systems and physics to engineering and economics has created a paradigm shift in the way systems are analyzed, designed, and built. Typical examples of such complex networks include metabolic pathways \cite{Goldbeter96}, vehicular platoons \cite{Jovanovic05,Jovanovic08,Raza96,Seiler04,Swaroop99}, arrays of micro-mirrors \cite{Neilson01}, micro-cantilevers \cite{Napoli99}, and smart power grids. These systems are diverse in their detailed physical behavior, however, they share an important common feature that all of them consist of an interconnection of a large number of systems that affect each others' dynamics. There have been some progress in characterization of fundamental limitations of feedback for some classes of dynamical networks. For example, only to name a few, reference \cite{Middleton10} gives conditions for string instability in an array of linear time-invariant autonomous vehicles with communication constraints,  \cite{Vinay07} provides a lower bound on the achievable quality of disturbance rejection using a decentralized controller for stable discrete time linear systems with time delays, \cite{Padmasola06} studies the performance of spatially invariant plants interconnected through a static network, \cite{Leong} studies the time domain waterbed effect for single state linear systems and shows time domain analysis is useful for understanding the waterbed effect with respect to $l_1$-norm optimal control, and \cite{Siami14arxiv} investigates performance deterioration in linear dynamical networks subject to external stochastic disturbances and quantifies several explicit inherent fundamental limits on the best achievable levels of performance and show that these limits of performance are emerged only due to the specific interconnection topology of the coupling graphs. Furthermore, \cite{Siami14arxiv} characterizes some of the inherent fundamental tradeoffs between notions of sparsity and performance in linear consensus networks.

	Most of the above cited research on fundamental limitations of feedback in interconnected dynamical systems have been focused on networks with linear time-invariant dynamics. The main motivation of this paper stems from a recent work presented in \cite{chandra11} that presents that glycolysis oscillation can be an indirect effect of fundamental tradeoffs in this system. The results of this work is based on a linearized model of a two-state model of glycolysis pathway and tradeoffs are stated using Bode's results. In this paper, our approach to characterize hard limits is essentially different in the sense that it uses higher dimensional and more detailed nonlinear models of the pathway. We interpret fundamental limitations of feedback by using hard limits (lower bounds) on $\mathcal L_2$-gain disturbance attenuation of the system \cite{Middleton03, Schwartz96, Schaft}, and $\mathcal L_2$-norm squared of the output of the system \cite{motee6,seron99}.

	In this paper, our goal is to build upon our previous results \cite{motee6, siami12} and develop methods to characterize hard limits on performance of autocatalytic pathways. First, we study the properties of such pathways through a two-state model, which obtained by lumping all the intermediate reactions into a single intermediate reaction (Fig. \ref{fig_glycolysis-1}). Then, we  generalize our results to autocatalytic pathways that are composed of a chain of enzymatically catalyzed intermediate reactions (Fig. \ref{fig_glycolysis-2}). We show that due to the existence of autocatalysis in the system (which is a biochemical necessity)
, a fundamental tradeoff between a notion of fragility and net production of the pathway emerges. Also, we show that as the number of intermediate reactions grows, the price for better performance increases. 


\section{Minimal Autocatalytic Pathway Model}
\subsection{Two-State Model}
\label{section2}
	{We consider autocatalysis mechanism in a glycolysis pathway. The central role of glycolysis is to consume glucose and produce adenosine triphosphate (ATP), the cell's energy currency. Similar to many other engineered systems whose machinery runs on its own energy product, the glycolysis reaction is autocatalytic. The ATP molecule contains three phosphate groups and energy is stored in the bonds between these phosphate groups. Two molecules of ATP are consumed in the early steps (hexokinase, phosphofructokinase/PFK) and four ATPs are generated as pyruvate is produced. PFK is also regulated such that it is activated when the adenosine monophosphate (AMP)/ATP ratio is low; hence it is inhibited by high cellular ATP concentration \cite{Goldbeter96,Selkov75}. This pattern of product inhibition is common in metabolic pathways. We refer to \cite{chandra11} for a detailed discussion. 

	Experimental observations in Saccharomyces cerevisiae suggest that there are two synchronized pools of oscillating metabolites \cite{Hynne01}. Metabolites upstream and downstream of phosphofructokinase (PFK) have $180$ degrees phase difference, suggesting that a two-dimensional model incorporating PFK dynamics might capture some aspects of system dynamics \cite{Betz65}, and indeed, such simplified models qualitatively reproduce the experimental behavior \cite{Goldbeter96,Selkov75}. }

We assume that a lumped variable $x$ can encapsulate relevant information of all intermediate metabolites and consider a minimal model with three biochemical reactions as follows
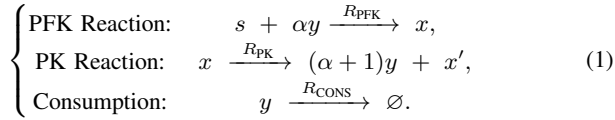
\begin{eqnarray}
	\begin{cases}
	\begin{matrix}
	\text{PFK Reaction:}& s ~+~ \alpha y  \xrightarrow{~R_{\text{PFK}}~} ~ x, \\
	\text{PK Reaction:}& x~ \xrightarrow{~R_{\text{PK}}~}  ~ (\alpha + 1) y ~+~ x',\\
	\text{Consumption:}& y   ~ \xrightarrow{~R_{\text{CONS}}~} ~ \varnothing.
	\end{matrix}
	\end{cases}
	\label{reaction-two}
\end{eqnarray}
In the PFK reaction, $s$ is some precursor and source of energy for the pathway with no dynamics associated, $y$ denotes the product of the pathway (ATP), $x$ is intermediate metabolites, $x'$ is one of the by-products of the second biochemical reaction (pyruvate kinase/PK). $\varnothing$ is a null state, $\alpha>0$ is the number of $y$ molecules that are invested in the pathway, and $\alpha + 1$ is the number of $y$ molecules produced. $A \xrightarrow{~k~} B$ denotes a chemical reaction that converts the chemical species $A$ to the chemical species $B$ at rate $k$.  The PFK reaction consumes $\alpha$ molecules of ATP with allosteric inhibition by ATP.
In the second reaction, pyruvate kinase (PK) produces $\alpha+1$ molecules of ATP for a net production of one unit\footnote{For the sake of simplicity of notations, we normalize the reactions such that consumption of one molecule of $y$ produces two molecules of $y$, which is equivalent to $\alpha=1$. }. The third reaction models the cell's consumption of ATP. We refer to Fig. \ref{fig_glycolysis-1} for a schematic diagram of biochemical reactions in the minimal model.

	A set of ordinary differential equations that govern the changes in concentrations $x$ and $y$ can be written as
\begin{eqnarray}
	\begin{cases}
	\dot{x} \, = \, R_{\text{PFK}}(y) \,-\, R_{\text{PK}}(x,y),\\
	\dot{y} \, = \, -\alpha \, R_{\text{PFK}}(y) \,+\, (\alpha+1) \, R_{\text{PK}}(x,y) \,-\, R_{\text{CONS}}(y).
	\end{cases}
	\label{model-2-s}
\end{eqnarray}
The reaction rates are chosen according to the following steps. For the PFK reaction, we have
\begin{equation}
	R_{\text{PFK}}(y)~=~\frac{2 y^{a}}{1+ y^{2h}},
	\label{eq-150}
\end{equation}
where $a$ models cooperativity of ATP binding to PFK and $h$ is the feedback strength of ATP on PFK. For the PK reaction, we use
\begin{equation}
	R_{\text{PK}}(x,y)~=~\frac{2 k x}{1+ y^{2g}},
\end{equation}
where $k$ is intermediate reaction rate and $g$ is the feedback strength of ATP on PK. The coefficients 2 in the numerator and feedback coefficient of the reaction rates come from the normalization. Finally,  the product $y$ is consumed by basal consumption rate of $1+\delta$, \ie
\begin{equation}
	R_{\text{CONS}}~=~1~+~ \delta
	\label{model-dist}
\end{equation}
in which $\delta$ is the perturbation in ATP consumption{\footnote{~In Example \ref{ex-2} of Section \ref{sec2}, the case of $R_{\rm CONS}~=~ k_y y + \delta$ is also studied.}.  In Section \ref{sec2}, we consider more general reaction rates which are suitable for a broad class of chemical kinetics models such as Michaelis-Menten and mass-action. Reaction rates (\ref{eq-150})-(\ref{model-dist}) are consistent with biological intuition and experimental data in the case of the glycolysis pathway \cite{chandra11}.
In the final reaction, the effect of an external time--varying disturbance $\delta$ on ATP demand is considered. 
The product of the pathway, ATP, inhibits the enzyme that catalyzes the first and second reactions, and the exponents $h$ and $g$ capture the strength of these inhibitions, respectively. By combing all steps, the nonlinear dynamics of  (\ref{model-2-s})-(\ref{model-dist}) can be cast as 
\begin{equation}
	\begin{cases}
	\dot{x}_1 ~ = ~ \frac{2 x_2^{a}}{1+ x_2^{2h}} ~-~ \frac{2 k x_1}{1+ x_2^{2g}},\\
	\\
	\dot{x}_2 ~ = ~ -\alpha \frac{2 x_2^{a}}{1+ x_2^{2h}} ~+~ (\alpha+1)\frac{2 k x_1}{1+ x_2^{2g}} ~-~ \left ({1+\delta}\right ), 
	\end{cases}
	\label{gly-two}
\end{equation}
with output variable
\begin{equation}
		y=x_2
\end{equation}
for $x_1,x_2 \geq 0$. 

In order to make several comparisons possible, we normalize all concentrations such that the equilibrium point of the unperturbed  system (i.e., when $\delta = 0$)  becomes 
\begin{equation}
	\left[\begin{array}{c}
	x_1^* \\
	x_2^*
	\end{array}\right]
	\,=\, \left[\begin{array}{c}
	\frac{1}{k} \\
	1
	\end{array}\right].
	\label{fixed-point}
\end{equation}
This can be achieved by nondimensionalizing the model.
%
%
%
%
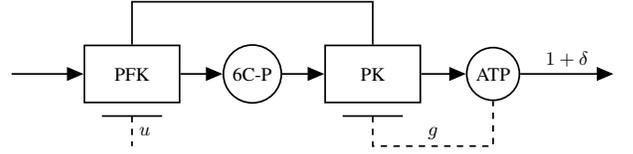
\begin{figure}
        \begin{center}
        \scalebox{.8}{
	\begin{tikzpicture}[auto, thick, node distance=2cm, >=triangle 45]
	\draw
	node at (0,0)[right=-3mm]{}
	node [input, name=input] {} 
	node [draw, thick, rectangle, minimum height = 3em,
    	minimum width = 5em, right of=input] (PFK) {PFK}
	node [sum, right of=PFK] (inter) {6C-P}
	node [draw, thick, rectangle, minimum height = 3em,
 	 minimum width = 5em, right of=inter] (PK) {PK}
	node [sum, right of=PK] (ATP) {ATP}
   	node [output, right of=ATP] (output) {} ;
	\draw[->](input) -- node {} (PFK);
 	\draw[->](PFK) -- node {} (inter);
	\draw[->](inter) -- node {} (PK);
	\draw[->](PK) -- node {} (ATP);
	\draw[->](ATP) -- node {$1+\delta$} (output);
    	\draw [-] (PK) -- (6,1.2)--(2,1.2)--(PFK);
    	\draw [-] (1.5,-.7) -- (2.5,-.7);
    	\draw [-] (5.5,-.7)--(6.5,-.7);
    	\draw [dashed] (6,-.7)--(6,-1.2)--node {$g$} (8,-1.2)--(ATP);
    	\draw [dashed] (2,-.7)--node {$u$} (2,-1.2);
	\end{tikzpicture}}
        \end{center}
        \caption{A schematic diagram of the minimal glycolysis model. The constant  glucose input along with $\alpha$ ATP molecules produce a pool of intermediate metabolites, which then produces $\alpha+1$ ATP molecules.}
        \label{fig_glycolysis-1}
\end{figure}

	In the minimal glycolysis model (\ref{gly-two}) expression $\frac{2}{1+ x_2^{2h}}$ can be interpreted as the effect of the regulatory feedback control mechanism employed by nature, which captures inhibition of the catalyzing enzyme. This observation suggests the following control system model for the minimal model of the glycolysis pathway
\begin{eqnarray}
	\hspace{-.5cm}\begin{bmatrix}
	\dot{x}_1\\
	\dot{x}_2
	\end{bmatrix} &=& \begin{bmatrix}
	{1}\\
	{-\alpha}
	\end{bmatrix} x_2^a u+
	\begin{bmatrix}
	{-1}\\
	{\alpha+1}
	\end{bmatrix} \frac{2 k x_1}{1+ x_2^{2g}}  \label{cont-gly1} - \begin{bmatrix}
	{0}\\
	{1+\delta}
	\end{bmatrix},
	\label{cont-gly-1} 
\end{eqnarray}
where $u$ is the control input and captures the effect of a general feedback control mechanism. Our primary motivation  behind development and analysis of such control system models for this metabolic pathways is to rigorously show that existing fundamental tradeoffs in such models are truly unavoidable and independent of control mechanisms used to regulate such pathways. For glycolysis autocatalytic pathways, the results of the following sections assert that the existing fundamental limits on performance  of the pathway depend only on  the autocatalytic structure of the underlying network.

{\it Stability properties of this model:}
According to \cite{chandra11}, the equilibrium point (\ref{fixed-point}) of two-state glycolysis model (\ref{gly-two}) is stable if 
{\[ 0~<~ h-a~<~\frac{k+g(1+\alpha)}{\alpha}.\] }
Our aim is to show that for any stabilizing control input there is a fundamental  limit on the best achievable performance by the closed-loop pathway.

\subsection{Performance Measures}
We quantify fundamental limits on performance of the glycolysis pathway via two different approaches.  

\subsubsection{$\mathcal{L}_2$-Gain from Exogenous Disturbance Input to Output}
	{In order to quantify lower bounds on the best achievable closed-loop performance of the two-state model (\ref{cont-gly1}), we need to solve the corresponding regional state feedback $\mathcal {L}_2$-gain disturbance attenuation problem with guaranteed stability. This problem consists of determining a control law $u$ such that the closed-loop system has the following properties: (i) the zero equilibrium of the system (\ref{cont-gly1}) with $\delta(t) = 0$ for all $ t \geq 0$ is asymptotically stable with region of attraction containing $\Omega$ (an open set containing the equilibrium point), (ii) 
for every $\delta \in {L}_2 (0, T )$ such that the trajectories of the system remain in $\Omega$, the $\mathcal {L}_2$-gain of the system from $\delta$ to $y$ is less than or equal to $\gamma$, \ie 
	\begin{equation}
		\int_{0}^{ T}(y(t)-y^*)^2 dt ~\leq ~\gamma^2 \int_{0}^{T}\delta^2(t) dt
		\label{problem}
	\end{equation}
for all $T\geq 0$ and zero initial conditions. 

It is well-known that there exists a solution to the static state feedback $\mathcal {L}_2$-gain disturbance attenuation problem with guaranteed stability, in some neighborhood of the equilibrium point, if there exists a smooth positive definite solution of the corresponding Hamilton-Jacobi inequality; we refer to \cite{Middleton03, Schaft} for more details.}

{The simplest robust performance requirement for model (\ref{cont-gly-1}) is that the concentration of $y$ (i.e., ATPs) remains nearly constant when there is a small constant disturbance in ATP consumption $\delta$ (see \cite{motee6,chandra11}). However, even temporary ATP depletion can result in cell death. Therefore, we are interested in a more complete picture of the transient response to external disturbances. We show that there exists a hard limit on the best achievable disturbance attenuation, which we denoted it by $\gamma ^{*}$, for system (\ref{cont-gly1}) such that the problem of disturbance attenuation (\ref{problem}) with internal stability is solvable for all $\gamma > \gamma^{*}$, but not for all $\gamma < \gamma^{*}$.
For a linear system, it is known that the optimal disturbance attenuation can be calculated using  zero-dynamics of the system \cite{Middleton03,seron99}. There is no fundamental limit on performance if and only if exogenous   disturbance $\delta$ does not influence the unstable part of the zero-dynamics of the system (as it is defined in \cite{Schwartz96} for nonlinear systems).}

\subsubsection{$\mathcal{L}_2$-Norm or Total Energy of the Output}
	{We characterize fundamental limitations of feedback for system (\ref{cont-gly1}) with initial condition $x(0)=x_0$ and zero external disturbances (i.e., $\delta(t)=0$) by considering the corresponding cheap optimal control problem. This case consists of finding a stabilizing state feedback control which minimizes the functional
	\begin{equation}
		J_{\epsilon}(x_0;u)~=~\frac{1}{2}~\int_{0}^{\infty}~\big[~\left( y(t)-y^* \right)^2 ~+~ \epsilon^2 \left(u(t)-u^*\right)^2 ~\big]~dt, \label{cheap-cost}
	\end{equation}
when $\epsilon$ is a small positive number. As $\epsilon \rightarrow 0$, the optimal value $J^*_{\epsilon}(x_0)$ tends to $J^*_0(x_0)$, the ideal performance of the system. It is well-known (e.g., see \cite{sepulchre97}, page $91$) that this problem has a solution if there exists a positive semidefinite optimal value function which satisfies the corresponding Hamilton--Jacobi-Bellman equation (HJBE). The interesting fact is that the ideal performance  is indeed a hard limit on performance of system (\ref{cont-gly1}). It is known that for a specific class of systems the ideal performance is the optimal value of the minimum energy problem for the zero-dynamics of the system (see \cite{seron99} for more details). The ideal performance (hard limit function) is zero if and only if the system has an asymptotically stable zero-dynamics subsystem.}

\subsection{Fundamental limits on the Performance Measures}

\subsubsection{$\mathcal L_2$-Gain Disturbance Attenuation}
In the following, it is shown that there exists a hard limit on the best achievable degrees of disturbance attenuation for system (\ref{cont-gly-1}). 

\begin{theorem}\label{theorem-01}
	Consider the optimal $\mathcal{L}_2$-gain disturbance attenuation problem for the minimal glycolysis model (\ref{cont-gly-1}). Then, the best achievable disturbance attenuation gain $\gamma ^{*}$ for system (\ref{cont-gly-1}) satisfies the following inequality
	\begin{equation}
		\gamma^* ~\ge ~\mathbf{\Gamma}(\alpha,k,g)
		\label{roboust}
	\end{equation}
and the hard limit function is given by
	\begin{equation}
		\mathbf{\Gamma}(\alpha,k,g)~=~\frac{\alpha}{k+g\alpha}.
	\end{equation}
\end{theorem}

\begin{proof}
	We recall that the optimal value of the achievable disturbance attenuation level $\gamma^*$ is a number with the property that the problem of disturbance attenuation with internal stability is locally solvable for each prescribed level of attenuation $\gamma > \gamma^*$ and not for $\gamma < \gamma^*$. 
In the first step, we introduce a new auxiliary variable $z=x_1+\frac{1}{\alpha}x_2$. By transforming the dynamics of the system using the following change of coordinates 
	\begin{equation}
        \left[\begin{array}{c}
        y \\
        z
        \end{array}\right]=\left[\begin{array}{cc}
        0 & 1 \\
        1 & \frac{1}{\alpha}
        \end{array}\right] \left[\begin{array}{c}
        x_1 \\
        x_2
        \end{array}\right],
	\end{equation}
we obtain the following form
	\begin{eqnarray}
		\begin{cases}
		\dot{y}~=~ - \frac{\alpha+1}{\alpha}\frac{2ky}{1+y^{2g}} + 
 (\alpha+1)\frac{2kz}{1+y^{2g}}-\alpha y^au-(1+\delta) \label{zero-gly0} \\
		\dot{z}~=~\frac{1}{\alpha}\frac{2 k z}{1+y^{2g}}-\frac{1}{\alpha^2}\frac{2 k y}{1+y^{2g}}-\frac{1}{\alpha}(1+\delta).
		\end{cases}
		\label{zero-gly}
	\end{eqnarray} 
%
Note that the optimal $\mathcal{L}_2$-gain disturbance attenuation of transformed system (\ref{zero-gly}) and the original system are the same. Based on \cite[Section 8.4]{Schaft2000} the optimal disturbance level for the linearized problem will provide a lower bound for the optimal disturbance of the nonlinear system. 
Furthermore, for the linear system this problem reduces to a disturbance attenuation problem for the zero dynamics with cost on the control input. Thus we consider the linearized zero dynamics of (\ref{zero-gly}) as follows
	\begin{equation}
		\dot{\bar z}~=~\frac{k}{\alpha}\bar z-\frac{g\alpha+k}{\alpha^2}\bar y-\frac{1}{\alpha}\delta, \label{linearized-1}
	\end{equation} 
where 
	\begin{eqnarray}
		\begin{cases}
		\bar z ~ = ~ z-z^* ~=~ z-(x^*+\frac{1}{\alpha}y^*) \\
		\bar y ~ = ~ y - y^*
		\end{cases}
	\end{eqnarray}
	We now calculate optimal disturbance attenuation problem (from $\delta$ to $y$) for the zero dynamics with cost on its control input $y$. For system (\ref{linearized-1}), the optimal value of $\gamma$ is given by (see \cite{Scherer92, Schwartz96} for more details)
	\begin{equation}
		\gamma_L^*~=~ \frac{\alpha}{k+g\alpha}.
		\label{gamma}
	\end{equation} 
Thus, we can conclude that 
\[\gamma^* ~\ge~ \gamma_L^*~=~\mathbf{H}(\alpha,k,g)~=~\frac{\alpha}{k+g\alpha}.\]
\end{proof}

	Theorem \ref{theorem-01} illustrates a tradeoff between robustness and efficiency (as measured by complexity and metabolic overhead). From (\ref{roboust}) the glycolysis mechanism is more robust efficient if $k$ and $g$ are large. On the other hand, large $k$ requires either a more efficient or a higher level of enzymes, and large $g$ requires a more complex allosterically controlled PK enzyme; both would increase the cell's metabolic load. 
The hard limit function $\mathbf{\Gamma}(\alpha,k,g)$ in Theorem \ref{theorem-01} is an increasing function of $\alpha$. This implies that increasing $\alpha$ (more energy investment for the same return) can result in worse performance. 
It is important to note that these results are consistent with results in \cite{chandra11}, where a linearized model with a different performance measure is used.

\subsubsection{Total Output Energy}
It is shown that there exists a hard limit on the best achievable ideal performance ($\mathcal{L}_2$-norm of the output) of system (\ref{cont-gly-1}). One can see that some minimum output energy (\ie ATP) is required to stabilize the unstable zero-dynamics (\ref{zero-gly}). This output energy represents the energetic cost of the cell to stabilize it to its steady-state. In the following theorem, we show that the minimum output energy is lower bounded by a constant which is only a function of the parameters and initial conditions of the glycolysis model. This hard limit is independent of the feedback control strategy used to stabilize the system.
\begin{theorem}\label{theorem-03}
	Suppose that the equilibrium of interest is given by (\ref{fixed-point}) and ${u^*}=1$. Then, there is a hard limit on the performance measure of the  unperturbed ($\delta = 0$) system (\ref{cont-gly-1}) in the following sense
	\begin{equation}
		\int_{0}^{\infty}~(y(t;u_0)-\bar{y})^2~dt~\geq~ \frac{\alpha^3k}{(g\alpha+k)^2}~z_0^2 +J(z_0;\alpha,k,g),
	\end{equation} 
where $z_0=\left(x(0)-{x^*}\right)+\frac{1}{\alpha}\left(y(0)-{y^*}\right)$, $u_0$ is an arbitrary stabilizing feedback control law for system (\ref{cont-gly-1}), $J(0;\alpha,k,g)=J(z;\alpha,k,0)=0$ and $|J(z;\alpha,k,g)| \leq c|z|^3$ on an open set $\Omega$ around the origin in $\R$.
\end{theorem}

\begin{proof}
By introduction of a new variable $z=x_1+\frac{1}{\alpha}y$, we rewrite (\ref{cont-gly-1}) in the canonical form (\ref{zero-gly}). 
We denote by $\pi(y,z;\epsilon)$ the solution of the HJB PDE corresponding to the cheap optimal control problem to (\ref{cont-gly-1}). We apply the power series method \cite{Albrecht61,Lukes69} by first expanding $\pi(y,z;\epsilon)$ in series as follows
	\begin{equation}
		\pi(y,z;\epsilon)~=~ \pi^{[2]}(y,z;\epsilon) ~+~ \pi^{[3]}(y,z;\epsilon) ~+~ \ldots \label{power_series}
	\end{equation}
in which $k$th order term in the Taylor series expansion of $\pi(y,z;\epsilon)$ is denoted by $\pi^{[k]}(y,z;\epsilon)$. Then (\ref{power_series}) is plug into the corresponding HJB equation of the optimal cheap control problem. The first term in the series is
	\begin{equation*} 
		\pi^{[2]}(y,z;\epsilon) ~=~ \left[
                               \begin{array}{cc}
                                 y-y^* & z-z^* \\
                               \end{array}
                             \right] P(\epsilon) \left[
                                                   \begin{array}{c}
                                                     y-y^*  \\
                                                     z-z^*  \\
                                                   \end{array}
                                                 \right],
	 \end{equation*}
where $P(\epsilon)$ is the solution of algebraic Riccati equation to the cheap control problem for the linearized model $(A_0,B_0)$. It can be shown that $P(\epsilon)$ can be decomposed in the form of a series in $\epsilon$ (see \cite{KwakernaakS72} for more details)
	\begin{equation*}
		P(\epsilon)~=~\left[
                  \begin{array}{cc}
                    \epsilon P_1 & \epsilon P_2 \\
                    \epsilon P_2 & P_0+\epsilon P_3 \\
                  \end{array}
                \right] + \mathcal{O}(\epsilon^2).
     \end{equation*}
Since the pole of the zero-dynamics of the linearized model is located at the $\frac{k}{\alpha}$, we can verify that $P_0 =\frac{2\alpha^3k}{(g\alpha+k)^2}$. Therefore, it follows that $\pi^{[2]}(y,z;\epsilon)=\frac{\alpha^3k}{(g\alpha+k)^2}z_0^2+\mathcal{O}(\epsilon)$. 
We only explain the key steps. One can obtain governing partial differential equations for the higher-order terms $\pi^{[k]}(y,z;\epsilon)$ for $k \geq 3$ by equating the coefficients of  terms with the same order. It can be shown that $\pi^{[k]}(y,z)=\pi_0^{[k]}(z) + \epsilon \pi_1^{[k]}(y,z)+\mathcal{O}(\epsilon^2)$ for all $k \geq 3$. 
Then, by constructing approximation of the optimal control feedback by using computed Taylor series terms, one can prove that $\pi(y,z;\epsilon) \rightarrow \frac{\alpha^3k}{(g\alpha+k)^2}z_0^2+(\text{higher order terms in } z_0)$ as $\epsilon \rightarrow 0$. Thus, the ideal performance cost value is $\frac{\alpha^3k}{(g\alpha+k)^2}z_0^2+J(z_0;\alpha,k,g)$.

\end{proof}

According to Theorems \ref{theorem-01} and \ref{theorem-03}, a fundamental tradeoff between a notion of fragility and net production of the pathway emerges as follows:  increasing $\alpha$ (number of ATP molecules invested in the pathway), increases fragility of the network to small disturbances (based on Theorem \ref{theorem-01}) and it can result in undesirable transient behavior (based on Theorem \ref{theorem-03}). 
	The large fluctuation in the level of ATP is not desirable, if the level of ATP drops below some threshold, there will not be sufficient supply of ATP for different pathways in the cell and that can result to cell death.  


	

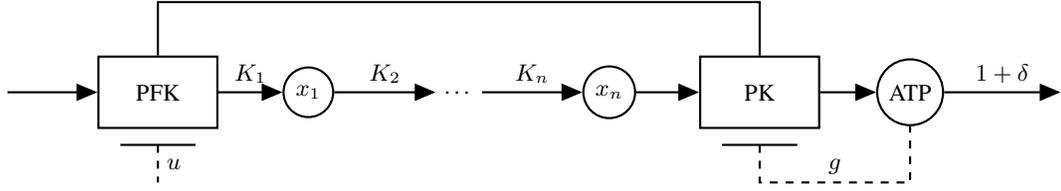
\begin{figure*}
        \begin{center}
        \scalebox{1}{
\begin{tikzpicture}[auto, thick, node distance=2cm, >=triangle 45]
\draw
	node at (0,0)[right=-3mm]{}
	node [input, name=input] {} 
	node [draw, thick, rectangle, minimum height = 3em,
    minimum width = 5em, right of=input] (PFK) {PFK}
	node [sum, right of=PFK] (x1) {$x_1$}
	node [right of=x1](noghte){\ldots}
	node [sum, right of=noghte] (xn) {$x_n$}
	node [draw, thick, rectangle, minimum height = 3em,
    minimum width = 5em, right of=xn] (PK) {PK}
	node [sum, right of=PK] (ATP) {ATP}
   	node [output, right of=ATP] (output) {};
	\draw[->](input) -- node {} (PFK);
 	\draw[->](PFK) -- node {$K_1$} (x1);
 	\draw[->](x1)-- node {$K_2$} (noghte);
 	\draw[->](noghte)-- node {$K_{n}$} (xn);
 	\draw[->](xn)-- node {} (PK);
	\draw[->](PK) -- node {} (ATP);
    \draw[->](ATP) -- node {$1+\delta$} (output);
    \draw [-] (PK) -- (10,1.2)--(2,1.2)--(PFK);
    \draw [-] (1.5,-.7) -- (2.5,-.7);
    \draw [-] (9.5,-.7)--(10.5,-.7);
    \draw [dashed] (10,-.7)--(10,-1.2)--node {$g$}(12,-1.2)--(ATP);
    \draw [dashed] (2,-.7)--node {$u$} (2,-1.2);
\end{tikzpicture}}
        \end{center}
        \caption{ A schematic diagram of a glycolysis pathway model with intermediate reactions. The constant  glucose input along with $\alpha$ ATP molecules produce a pool of intermediate metabolites, which then produces $\alpha+1$ ATP molecules.}
        \label{fig_glycolysis-2}
	\end{figure*}

\section{Autocatalytic Pathways With Multiple Intermediate Metabolite Reactions}
In Subsection {\ref{section2}}, we studied the property of such pathways with a two-state model (\ref{cont-gly-1}), which is obtained by lumping all the intermediate reactions into a single intermediate reaction.  In the next step, we consider autocatalytic pathways with multiple intermediate metabolite reactions as shown in Fig. \ref{fig_glycolysis-2}:
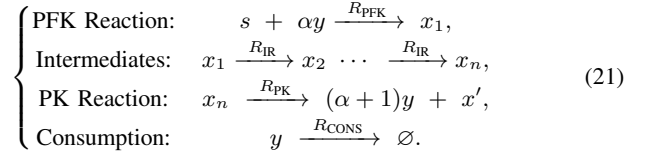
	\begin{eqnarray}
        \begin{cases}
        \begin{matrix}
        \text{PFK Reaction:}& s ~+~ \alpha y  \xrightarrow{~R_{\text{PFK}}~} ~ x_1, \\
        \text{Intermediates:}& x_1 \xrightarrow{~R_{\text{IR}}~} x_2 ~  \cdots   ~\xrightarrow{~R_{\text{IR}}~}x_{n},\\
        \text{PK Reaction:}& x_{n}~ \xrightarrow{~R_{\text{PK}}~}  ~ (\alpha + 1) y ~+~ x',\\
        \text{Consumption:}& y   ~ \xrightarrow{~R_{\text{CONS}}~} ~ \varnothing.
        \end{matrix}
        \end{cases}
        \label{reaction-n}
	\end{eqnarray}
A set of ordinary differential equations that govern the changes in concentrations of $x_i$ for $i=1, \ldots, n$ and $y$ can be obtained as follows
        \begin{eqnarray}
       	 \begin{cases}
        	\dot{x}_1 ~=~ R_{\text{PFK}}(y) \,-\, R_{\text{IR}}(x_1),
        	\\ \label{cont-gly3}
         	\dot{x}_2~ =~ R_{\text{IR}}(x_1)\,-\,R_{\text{IR}}(x_2),
        	\\
        	~~~~~ \vdots \\
         	\dot{x}_n~=~ R_{\text{IR}}(x_{n-1})\, - \, R_{\text{PK}}(x_n,y),  \\
         	~\dot{y}  ~ = ~(\alpha+1)R_{\text{PK}}(x_n,y)\, -\, \alpha R_{\text{PFK}}(y) \,-\, {R_{\text{CONS}}}
         	\end{cases}
        	\label{general-model}
        \end{eqnarray}
for $x_i \geq 0$ and $y \geq 0$. Our notations are similar to those of the two-state pathway model (\ref{reaction-two}). The reaction rates are given as follows
	\begin{eqnarray}
	\begin{cases}
		R_{\text{PFK}}(y)~=~\frac{2 y^{a}}{1+ y^{2h}},\\
		R_{\text{PK}}(x_n,y)~=~\frac{2 K_n x_n}{1+ y^{2g}},\\
		R_{\text{IR}}(x_i)~=~K_i x_i ~~\text{for}~~n=1,2,\ldots,n,\\
		R_{\text{CONS}}~=~1+\delta
		\label{rates}
		\end{cases}
	\end{eqnarray}
Furthermore, in the glycolysis model (\ref{general-model}), similar to the minimal model (\ref{cont-gly-1}), expression $\frac{2}{1+x^{2h}}$ can be interpreted as the effect of the regulatory feedback control mechanism employed by nature that captures inhibition of the catalyzing enzyme. Hence, we can derive a control system model for the autocatalytic pathway with multiple intermediate metabolite reactions as follows
\begin{small}
	\begin{equation}
        \begin{cases}
        \dot{x}_1 ~=~ y^a u~-~K_1x_1,\\
         \dot{x}_2~ =~ K_1{x_1}~ - ~ K_2x_2,\\
        ~~~~~ \vdots\\
         \dot{x}_n~=~ K_{n-1}{x_{n-1}}~ - ~ \frac{2K_n x_n}{1+y^{2g}},  \\
         \dot{x}_{n+1} ~ = ~(\alpha+1)\frac{2K_n x_n}{1+x_{n+1}^{2g}} ~ - ~ \alpha x_{n+1}^a  u ~- ~ {(1+\delta)},\\
         ~y~=~x_{n+1},
         \end{cases}
          \label{cascade-model} 
        \end{equation}
        \end{small}for $x_i \geq 0$ and $y \geq 0$. 
        In order to simplify our analysis and be able to calculate explicit formulae, we assume that $K:=K_1=\dots=K_n>0$. We normalize all concentrations such that unperturbed steady states become
	\begin{eqnarray}
        		y^*=x^*_{n+1}=1 ~~~\textrm{and}~~~  x_i=K^{-1}
        		\label{eq-point}
	\end{eqnarray}
for all $i=1,\ldots,n$. 





\subsection{$\mathcal L_2$-Gain Disturbance Attenuation}

	We extend our results in Theorem \ref{theorem-01} to  higher dimensional model of autocatalytic pathways. In the following theorem, we show that there exists a size-dependent hard limit on the best achievable disturbance attenuation for system (\ref{cascade-model}).

\begin{theorem}\label{theorem-02}
	Consider the optimal $\mathcal{L}_2$-gain disturbance attenuation problem for glycolysis model (\ref{cascade-model}). Then, the best achievable disturbance attenuation gain $\gamma ^{*}$ for system (\ref{cascade-model}) satisfies the following inequality
	\begin{equation}
		\gamma^* ~ \ge ~ \mathbf{\Gamma}(\alpha,K,g,n),
		\label{roboust-2}
	\end{equation}
where the hard limit function is given by
	\begin{eqnarray}
		&& \mathbf{\Gamma}(\alpha,K,g,n) = \nonumber  \\
		&&~~~~~~~\left[\left( K+g\alpha \left(\frac {\alpha+1}{\alpha}\right)^{\frac{n-1}{n}} \right) \left( \left(\frac {\alpha+1}{\alpha} \right)^{\frac{1}{n}}-1\right) \right]^{-1}. 
		\nonumber
\end{eqnarray}
\end{theorem}

\begin{proof}
	First, by  introducing a new variable $z_1=x_1+\frac{1}{\alpha}y$, we can cast the zero-dynamics of (\ref{cascade-model}) in the following form
	\begin{small}
	\begin{eqnarray}
		\begin{cases}
		\dot{z}_1 ~=~ -Kz_1~+~ \frac{\alpha+1}{\alpha}\frac{2Kx_n}{1+y^{2g}}~+~\frac{K}{\alpha}y-\frac{1}{\alpha}(\delta+1),
		 \\ 
		\dot{x}_2 ~=~ K{z_1}~-~\frac{K}{\alpha}y~-~Kx_2,
		\\
		~~~\cdots\\
		\dot{x}_n~=~ K{x_{n-1}} ~-~ \frac{2Kx_n}{1+y^{2g}}. 
		\end{cases}
		 \label{zero-gly-cascade} 
	\end{eqnarray}
	\end{small}
Let us define 
\begin{equation}
z := \begin{bmatrix}z_1& x_2& \dots& x_n\end{bmatrix}^{\text{T}},
\label{eq:667}
\end{equation}
and 
\begin{equation}
z^*:=\begin{bmatrix}\frac{1}{K}+\frac{1}{\alpha}& \frac{1}{K}& \ldots& \frac{1}{K}\end{bmatrix}^{\text T}.
\label{eq:672}
\end{equation}
Then, we rewrite (\ref{zero-gly-cascade}) in the following form
	\begin{equation}
		\dot{\bar z}~=~A\bar z~+~B\bar y~+~C\delta~+~\bar f(\bar z,\bar y),
		\label{nonlinear-zero-cas}
	\end{equation} 
where
	\begin{eqnarray}
		&&A~=~\left[
          \begin{smallmatrix}
              -K & 0& 0& ~\ldots~ & (1+\frac{1}{\alpha})K \\
               K &-K& 0& ~\ldots~ & 0 \\
               0 & K& -K&~\ldots~&0\\
               & \vdots&  & ~\ddots~ &  \vdots  \\
              0 & 0 & 0 & \ldots & -K \\
          \end{smallmatrix}
          \right],\nonumber  \\
          &&B~=~\left[
             \begin{smallmatrix}
               -\frac{\alpha+1}{\alpha}g+\frac{K}{\alpha} \\
              -\frac{K}{\alpha}  \\
              \vdots  \\
              g\\
            \end{smallmatrix}\right],~C~=~\left[
             \begin{smallmatrix}
              -\frac{1}{\alpha}  \\
              0\\
              \vdots  \\
              0\\
            \end{smallmatrix}\right],
            \label{A-B-C}
          \end{eqnarray}
$\bar z=z-z^*$, $\bar y= y - y^*$, $\bar f(0,0) = 0$ and
	\begin{equation}
		\Big {\|} \frac{\partial \bar f(\bar z,\bar y)}{\partial (\bar z,\bar y)} \Big{\|} ~\leq~ c |(\bar z,\bar y)|,
	\end{equation} 
near the origin in $\R^n$ for $c > 0$. 
Now, according to \cite{Schaft} we know that if the system (\ref{nonlinear-zero-cas}) has $\mathcal L_2$-gain $\leq \gamma$, then the linearized system has $ \mathcal L_2$-gain $\leq \gamma$.
Hence, we only consider the linearized system, \ie
	\begin{equation}
		\dot{\bar z}~=~A\bar z~+~B \bar y~+~C\delta.
		\label{zero-cas}
	\end{equation}
Note that $\lambda=K\left[(\frac{\alpha+1}{\alpha})^{\frac{1}{n}}-1\right]$ is the eigenvalue of $A$ with the greatest real part. And the corresponding left eigenvector of $\lambda$, is $v=\begin{bmatrix}1&  (\frac{\alpha+1}{\alpha})^{\frac{1}{n}}& \ldots& (\frac{\alpha+1}{\alpha})^{\frac{n-1}{n}}\end{bmatrix}^{\text T}$. Now, we consider the following subsystem of (\ref{zero-cas})
	\begin{equation*}
		\dot{\tilde z} = \lambda \tilde z+\Big[\big((1+\frac {1}{\alpha})^{\frac{n-1}{n}}-(1+\frac{1}{\alpha})\big)g-\frac{K}{\alpha}\big( (1+\frac {1}{\alpha})^{\frac{1}{n}} -1\big )\Big ]\bar y-\frac{1}{\alpha}\delta.
	\end{equation*}
Based on the result of \cite{Scherer92} and \cite{Schwartz96}, the formula to compute the optimal value of $\gamma$ reduces to
	\begin{equation*}
		\gamma_L^*~\geq~\frac{1}{\big( K+g\alpha (1+\frac {1}{\alpha})^{\frac{n-1}{n}} \big )\big((1+\frac {1}{\alpha})^{\frac{1}{n}}-1\big)}.
\end{equation*} 
Note that according to Proposition $6$ of \cite{Schaft}, $\gamma_L^*$ is a lower bound for the optimal $\gamma^*$ of the nonlinear system (\ref{cascade-model}).  
\end{proof}

\subsection{Total Output Energy}
It is proven that there exists a size-dependent hard limit on the best achievable ideal performance of system (\ref{cascade-model}).
	
\begin{theorem}\label{theorem-04}
		Suppose that the equilibrium of interest is given by (\ref{eq-point}) and ${u^*}=1$. Then, the $\mathcal{L}_2$-norm of the output  of the unperturbed system (\ref{cascade-model}) cannot be made arbitrarily small, which implies that there is a fundamental limit on performance in the following sense
	\begin{eqnarray}
		&& \hspace{-1.3cm} \int_{0}^{\infty}\big(y(t;u_0)-{y^*}\big)^2~dt~\\ \nonumber
		&&~~~~~~~~~~~~\geq~ \mathbf H(z_0;\alpha,K,g,n)~+J(z_0;\alpha,K,g,n),
	\end{eqnarray} 
where
	\begin{eqnarray}
		&&\hspace{-.5cm}\mathbf H(z_0;\alpha,K,g,n)= \nonumber\\ 
 		&&\frac{\alpha^2K\big ( \frac{1}{\alpha}(y(0)-y^*)+\sum_{i=1}^{n}(\frac{\alpha+1}{\alpha})^{\frac{i-1}{n}}(x_i(0)-x_i^*)\big)^2}{\big((\frac {\alpha+1}{\alpha})^{\frac{1}{n}}-1\big)\big( K+g\alpha (\frac {\alpha+1}{\alpha})^{\frac{n-1}{n}} \big )^2},\nonumber
	\end{eqnarray} 
$u_0$ is an arbitrary stabilizing feedback control law for system (\ref{cascade-model}), $z_0= z(0)-z^*$ where $z$ and $z^*$ are defined by \eqref{eq:667} and \eqref{eq:672} respectively, $J(0;\alpha,K,g,n)=J(z;\alpha,K,0,n)=0$, and $|J(z;\alpha,K,g,n)| \leq c|z|^3$ on an open set $\Omega$ around the origin in $\R^n$.
\end{theorem}

\begin{proof}
	The proof of this theorem based on results  from \cite{Albrecht61,Lukes69} and Theorem \ref{theorem-03}. Similar to the proof of Theorem \ref{theorem-02}, one can cast the zero-dynamics of the unperturbed system \eqref{cascade-model} as follows 
	\begin{equation}
		\dot{\bar z}~=~A\bar z~+~B\bar y~+~\bar f(\bar z,\bar y),
		\label{nonlinear-zero-cas}
	\end{equation} 
where $A$ and $B$ are given by \eqref{A-B-C},
$\bar z=z-z^*$, $\bar y= y - y^*$, $\bar f(0,0) = 0$ and
	\begin{equation*}
		\Big {\|} \frac{\partial \bar f(\bar z,\bar y)}{\partial (\bar z,\bar y)} \Big{\|} ~\leq~ c |(\bar z,\bar y)|
	\end{equation*} 
near the origin in $\R^n$ for $c > 0$. We denote by $\pi(y,z;\epsilon)$ the solution of the HJB PDE corresponding to the cheap optimal control problem to the unperturbed system (\ref{cascade-model}). We apply the power series method \cite{Albrecht61,Lukes69} by first expanding $\pi(y,z;\epsilon)$ in series as in \eqref{power_series}, where $\pi^{[k]}(y,z;\epsilon)$ denotes $k$'th order term in the Taylor series expansion of $\pi(y,z;\epsilon)$. Then, (\ref{power_series}) is plugged into the corresponding HJB equation of the optimal cheap control problem. The first term in the series is
	\begin{equation*} 
		\pi^{[2]}(y,z;\epsilon) ~=~ \left[
                               \begin{array}{cc}
                                 y-y^* & z-z^* \\
                               \end{array}
                             \right] P(\epsilon) \left[
                                                   \begin{array}{c}
                                                     y-y^*  \\
                                                     z-z^*  \\
                                                   \end{array}
                                                 \right],
	 \end{equation*}
where $P(\epsilon)$ is the solution of algebraic Riccati equation to the cheap control problem for the linearized model. It can be shown that $P(\epsilon)$ can be decomposed in the form of a series in $\epsilon$ (see \cite{KwakernaakS72} for more details)
	\begin{equation*}
		P(\epsilon)~=~\left[
                  \begin{array}{cc}
                    \epsilon P_1 & \epsilon P_2 \\
                    \epsilon P_2 & P_0+\epsilon P_3 \\
                  \end{array}
                \right] + \mathcal{O}(\epsilon^2)
     \end{equation*}
in which $P_0$ is the positive solution of the associated algebraic Riccati equation for $(A,B)$, \ie
\[ A^{\text{T}} P_0 + P_0 A~=~ P_0 BB^{\text{T}}P_0.\]
It follows that 
\[\pi^{[2]}(y,z;\epsilon)=\frac{1}{2}z_0^{\text{T}} P_0 z_0+\mathcal{O}(\epsilon).\]
 One can obtain governing partial differential equations for the higher-order terms $\pi^{[k]}(y,z;\epsilon)$ for $k \geq 3$ by equating the coefficients of  terms with the same order. It can be shown that 
 \[ \pi^{[k]}(y,z)=\pi_0^{[k]}(z) + \epsilon \pi_1^{[k]}(y,z)+\mathcal{O}(\epsilon^2)\] for all $k \geq 3$.  Then, by constructing approximation of the optimal control feedback by using computed Taylor series terms, one can prove that $\pi(y,z;\epsilon) \rightarrow \frac{1}{2}z_0^{\text{T}} P_0 z_0+(\text{higher order terms in } z_0)$ as $\epsilon \rightarrow 0$. Thus, the ideal performance cost value can be written as  
 \begin{equation}
 \lim_{\epsilon \rightarrow 0}\pi(y,z;\epsilon)~=~\frac{1}{2}z_0^{\text{T}} P_0 z_0+J(z_0;\alpha,K,g,n).
 \label{eq:835}
 \end{equation}
Next, we obtain a lower bound on $\frac{1}{2}z_0^{\text{T}} P_0 z_0$. The characteristic equation  of matrix $A$ is characterized by 
\[ (x+K)^{n}-\frac{\alpha +1}{\alpha}K^n~=~0.\]
 Therefore, one can see that $\lambda=K\left[(\frac{\alpha+1}{\alpha})^{\frac{1}{n}}-1\right]$ is the eigenvalue of $A$ with the greatest real part and its corresponding left eigenvector is $v=\begin{bmatrix}1&  (\frac{\alpha+1}{\alpha})^{\frac{1}{n}}& \ldots& (\frac{\alpha+1}{\alpha})^{\frac{n-1}{n}}\end{bmatrix}^{\text T}$. Now, let us consider the subsystem associated to this mode as follows
	\begin{equation*}
		\dot{\tilde z} = \lambda \tilde z+\Big[\big((\frac {\alpha+1}{\alpha})^{\frac{n-1}{n}}-(\frac{\alpha+1}{\alpha})\big)g-\frac{K}{\alpha}\big( (\frac {\alpha+1}{\alpha})^{\frac{1}{n}} -1\big )\Big ] \bar y,
	\end{equation*}
where 
\begin{equation}
 \tilde z~=~ v^{\text{T}} \bar z ~=~ \frac{1}{\alpha} \bar y + \sum_{i=1}^n \left( \frac{\alpha+1}{\alpha}\right)^{\frac{i-1}{n}} \bar x.
 \label{eq:847}
 \end{equation}
The corresponding cost value for this subsystem is given by  
\begin{equation}
 \frac{1}{2}z_0^{\text{T}} P_0 z_0 ~\geq~  \frac{\alpha^2K \tilde z(0)^2}{\big((\frac {\alpha+1}{\alpha})^{\frac{1}{n}}-1\big)\big( K+g\alpha (\frac {\alpha+1}{\alpha})^{\frac{n-1}{n}} \big )^2}
 \label{eq:853}
 \end{equation}
which is a lower bound for the linearized cost $ \frac{1}{2}z_0^{\text{T}} P_0 z_0$.
 Finally, using \eqref{eq:835}, \eqref{eq:847} and \eqref{eq:853}, we get the desired result. 
 
\end{proof}

In the case that the number of intermediate reactions is one (\ie $n=1$) the results of Theorems \ref{theorem-02} and \ref{theorem-04} reduce to the results of Theorems \ref{theorem-01} and \ref{theorem-03}, respectively.
%
%
Through a straightforward analysis, one can argue that $\mathbf H(z_0;\alpha,K,g,n) \in \mathcal{O}(n)$ and $\mathbf{\Gamma}(\alpha,K,g,n) \in \mathcal{O}(n)$, and they can be approximated by
	\begin{eqnarray}
		&&\hspace{-1cm}\mathbf H(z_0;\alpha,K,g,n) \approx \nonumber \\
		&&~\frac{\alpha^2K\big ( \frac{1}{\alpha}(y(0)-y^*)+{\sum_{i=1}^{n}(\frac{\alpha+1}{\alpha})^{\frac{i-1}{n}}(x_i(0)-x_i^*)}\big)^2}{\big( K+g(\alpha+1)  \big )^2 \ln(\frac{\alpha+1}{\alpha})}n\nonumber
	\end{eqnarray} 
and
	\begin{equation}
		\mathbf{\Gamma}(\alpha,K,g,n) ~\approx~ \frac{n}{ \big (g(\alpha+1)+K \big )\ln (1+\frac{1}{\alpha})}.
		\label{33}
	\end{equation}	
This implies that as the number of intermediate reactions $n$ grows, the price paid for robustness for both $\mathbf H(z_0;\alpha,K,g,n)$ and $\mathbf{\Gamma}(\alpha,K,g,n)$ increases linearly by network size $n$. In general, the larger the number of intermediate reactions involved in the breakdown of a metabolite, the less complex the enzymes involved in the individual reactions need to be. On the other hand, increasing the number of intermediate metabolites results in larger $\mathbf {\Gamma}$ and $\mathbf H$ which means less robustness to disturbances and having undesirable transient behavior.


\section{General Autocatalytic Pathways}
\label{sec2}
In the final step, we turn our focus on networks with autocatalytic structures (as shown in Fig. \ref{fig_3}) that  belong to a class of nonlinear dynamical networks with cyclic feedback structures driven by {disturbance}. Each network consists of a group of nonlinear subsystems with state-space dynamics
\begin{eqnarray}
&&\left\{ \begin{array}{rcl}
     \dot x_i &=&- f_i (x_i) + u_i \\
	y_i& =& g_i(x_i)
	\label{eq7}	
\end{array}\right.
\end{eqnarray}	
for $x_i \geq 0$, $y_i \geq 0$, $1\leq i \leq n$, and               
\begin{eqnarray}	               
&&\left\{ \begin{array}{rcl}
\dot x_{n+1} &=& -f_{n+1}(x_{n+1}) + u_{n+1}  - \alpha u, \\
y_{n+1} &=&  u, 
                \end{array}\right.
                	\label{eq8}
\end{eqnarray}
where $f_i(\cdot)$ and $g_i(\cdot)$ for $i=1,\ldots,n$ are increasing functions. Moreover,   $u_i(t)$, $y_i(t)$ and $x_i(t)$ are input, output and state variables of each subsystem, respectively. 
These assumptions are suitable for a broad class of chemical kinetics models such as Michaelis-Menten and mass-action. The state-space representation of the nonlinear cyclic interconnected network shown in Fig. \ref{fig_3} is given by 
\begin{eqnarray}
\begin{cases}
	\dot x_1 ~=~ -f_1 (x_1) + y_{n+1}, \\ 
	\dot x_2 ~=~ -f_2 (x_2) + y_1,\\
			~~~~\cdots~  \\
	\dot x_{n+1} ~=~ -f_{n+1} (x_{n+1})+ y_{n} - \alpha u +\delta,\\
	y ~=~ x_{n+1}. 
	\end{cases}	
\label{eq88}
\end{eqnarray}

\begin{assumption}
We assume that $x_i^*$ for $i=1,\ldots,n$ and $y^*$ are equilibrium points of the unperturbed system (\ref{eq88}). Moreover, it is assumed that 
\begin{equation}
a:=f^{\prime}_1(x_1^*)=f^{\prime}_2(x_2^*)= \cdots =f^{\prime}_{n}(x_n^*),
\label{a}
\end{equation}
where $f^{\prime}_i(x_i^*):=\left .\frac{\mathrm d f_i}{\mathrm d x_i} \right |_{x_{i}=x_i^*}$.
\end{assumption}
%
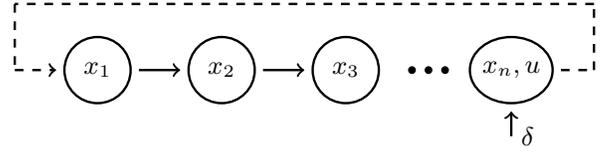
\begin{figure}[t]
        \begin{center}
        \scalebox{1.1}{
\begin{tikzpicture}
\draw [ thick] (-2.5,0) circle [radius=0.4];
\draw [ thick] (-1,0) circle [radius=0.4];
\draw [ thick] (.5,0) circle [radius=0.4];
\draw [  thick] (2.5,0) ellipse (.5 and .4);
\draw [->, thick] (-2,0) -- (-1.5,0);
\draw [->,  thick] (-.5,0) -- (0,0);
\draw [fill] (1.3,0) circle [radius=.04];
\draw [fill] (1.5,0) circle [radius=.04];
\draw [fill] (1.7,0) circle [radius=.04];
\draw [ thick, dashed] (3.1,0) -- (3.5,0);
\draw [ thick, dashed] (3.5,0) -- (3.5,.8);
\draw [ thick, dashed] (3.5,.8) -- (-3.5,.8);
\draw [ thick, dashed] (-3.5,.8) -- (-3.5,0);
\draw [ ->,thick, dashed] (-3.5,0) -- (-3,0);
\node[] at (-2.5,0) {$x_1$};
\node[] at (-1,0) {$x_2$};
\node[] at (.5,0) {$x_3$};
\node[] at (2.5,0) {$x_n,u$};
\draw [ ->,thick] (2.5,-.8) -- (2.5,-.5);
\node[] at (2.7,-.8) {$\delta$};
\end{tikzpicture}}
\end{center}
\caption{{\small The schematic diagram of the nonlinear network (\ref{eq88}) with a cyclic feedback structure with an output disturbance $\delta$ and control input $u$.}}
  	\label{fig_3}
\end{figure}
 
\begin{theorem}
\label{theorem-5}
For cyclic networks (\ref{eq88}), if 
\begin{equation}
r:= \left (\frac{g^{\prime}_1(x_1^*) g^{\prime}_2(x_2^*) \cdots g^{\prime}_{n}(x_n^*)}{\alpha}\right )^{\frac{1}{n}} > a,
\label{eqr}
\end{equation}
 then there exists a hard limit on the best achievable disturbance attenuation (i.e., $\gamma ^{*} >0$)  for system (\ref{eq88}) such that the regional state feedback $\mathcal{L}_2$--gain disturbance attenuation problem with stability constraint is solvable for all $\gamma >  \gamma^{*}$ and is not solvable for all $\gamma < \gamma^{*}$.
Furthermore, the hard limit function is given by
\begin{equation}
\gamma^* \ge \mathbf{\Gamma}(f^{\prime}_{n+1}(y^*),r,a)=\frac{1}{f^{\prime}_{n+1}(y^*)+r-a}.
\label{roboust-3}
\end{equation}
\end{theorem}

\begin{proof}
In the first step, we introduce a new auxiliary variable $z_1=x_1+\frac{1}{\alpha} x_{n+1}$.  We can cast the linearized zero-dynamics of (\ref{eq8}) in the following form
\begin{equation}
\hspace{-0.65cm} \dot{z} ~= ~A_0 z~+~B_0 y~ +~ C_0 \delta,\label{zero-gly-cascade-2} 
\end{equation}
where $z=[z_1,x_2,\cdots,x_{n}]^{\T}$,
\begin{eqnarray}
	&&A_0=\left[
		\begin{matrix}
			-a & 0  & \ldots & 0 & {\alpha}^{-1} g^{\prime}_{n}(x_n^*) \\
			g^{\prime}_1(x_1^*) & -a & \ldots & 0 & 0 \\
			\vdots & \vdots  & \ddots &\vdots  & \vdots &  \\
			0 & 0 & \ldots & -a & 0 \\
			0 & 0 & \ldots & g^{\prime}_{n-1} (x_{n-1}^*)& -a 
		\end{matrix}
		\right],\nonumber  \\
	&&B_0=\left[
		\begin{matrix}
			\frac{a-f^{\prime}_{n+1}(y^*)} {\alpha} \\
			-\frac{g^{\prime}_1(x_1^*)} {\alpha}  \\
			\vdots  \\
			0\\
			0
		\end{matrix}
		\right],~\text{and}~~C_0=\left[
		\begin{matrix}
			{\alpha}^{-1} \\
			0\\
			\vdots  \\
			0\\
			0
		\end{matrix}
		\right].
\end{eqnarray}
Then, we consider the characteristic equation of matrix $A_0$ which is given by 
\begin{equation}
 (\lambda+a)^{n}-r^{n}~=~0.
\label{charc}
 \end{equation}
From (\ref{eqr}) and (\ref{charc}), it follows that $\lambda_1=r-a$ is the eigenvalue of $A_0$ with the largest real-part value with left eigenvector  
\[v_1=\Big [\,1~,~  \frac{r}{g^{\prime}_1(x_1^*)}~, ~\ldots~,~ \frac{r^{n-1}}{g^{\prime}_1(x_1^*)g^{\prime}_2(x_2^*)\cdots g^{\prime}_{n-1}(x_{n-1}^*)}\,\Big]^{\T}.\] 
The unstable subsystem of (\ref{zero-gly-cascade-2}) is characterized by 
\begin{equation}
\dot{z}~=~\lambda_1 z
\, + \, {\alpha}^{-1}  \left (a-f^{\prime}_{n+1}(y^*)\, - \, r \right  )y \, + \, {\alpha}^{-1} \delta.
\label{zero-case}
\end{equation}
From the results of \cite{Scherer92} and \cite{Schwartz96}, the formula to compute the optimal value of $\gamma$ reduces to
\begin{eqnarray}
\gamma^*_L = \frac{1}{f^{\prime}_{n+1}(y^*)+r-a}.
\end{eqnarray} 
We emphasize that according to \cite[Proposition $6$]{Schaft}, $\gamma_L^*$ is a lower bound for the optimal $\gamma^*$ for the nonlinear system (\ref{eq88}).  
\end{proof}


\section{Examples} 
We apply our results to metabolic pathway \eqref{reaction-two} and quantify its existing hard limits. We assume that the second reaction  in \eqref{reaction-two} has no ATP feedback ATP on PK, \ie $g=0$. We consider two scenarios for the consumption rate $R_{\rm CONS}$; in the first example, we assume the product $y$ is consumed by basal consumption rate  $1+\delta$, and then, in the second example, we consider the case where the consumption rate depends on $y$. 

\begin{example}
Let us consider the minimal representation of autocatalytic glycolysis pathway given by \eqref{reaction-two}. It is assumed that the second reaction  in \eqref{reaction-two} has no ATP feedback ATP on PK, \ie $g=0$. Then, we can rewrite \eqref{gly-two} as follows 
\begin{eqnarray}
	\dot{x}_1 & = & \frac{2 y^{a}}{1+ y^{2h}} ~-~ k x_1,\\
	\dot{y} & = &-\alpha \frac{2 y^{a}}{1+ y^{2h}} ~+~ (\alpha+1) k x_1 ~-~ (1+\delta),
	\label{gly-two-ex}
\end{eqnarray}
for $x_1 \geq 0$ and $y \geq 0$. 
By considering expression $\frac{2 y^{a}}{1+ y^{2h}}$ as the regulatory feedback control employed by nature that captures inhibition of the catalyzing enzyme, a control system model for glycolysis can be obtained as follows
\begin{eqnarray}
\dot{x}_1 & = & -k \,x_1  +  u,  \label{cont-glyex1}\\
\dot{y} & = &  (\alpha+1) k \, x_1 -  {\alpha}\, u -1 -\delta, \label{cont-glyex2}
\end{eqnarray}
where $u$ is the control input. Using \eqref{cont-glyex1}-\eqref{cont-glyex2} and Theorem \ref{theorem-5}, it follows that 
 \begin{equation}
\gamma ~>~ \frac{\alpha}{k},
\label{ex-gain}
\end{equation}
 where the equilibrium point of the unperturbed system is given by $x_1=1/k$ and $y=1$.   As we expected \eqref{ex-gain} is consistent with the result of Theorem \ref{theorem-01}. 
 \end{example}
 
 \begin{example}
 \label{ex-2}
Let us now consider the minimal representation of autocatalytic glycolysis pathway represented by \eqref{reaction-two} with consumption rate depending on $y$ that is given by
  \[ R_{\rm CONS}~=~ k_y y + \delta.\]
We refer to \cite{BuziD10} for a complete discussion. 
 Then, a set of ordinary differential equations that govern the changes in concentrations $x_1$ and $y$ can be written as
\begin{eqnarray*} 
 	\dot{x}_1 & = &-k\, x_1 ~+~ \frac{2 y^{a}}{1+ y^{2h}}, \\
	\dot{y} & = & -\alpha \frac{2 y^{a}}{1+ y^{2h}} ~+~ (\alpha+1)k \, x_1~-~ \left ({k_y y+\delta}\right ),
 \end{eqnarray*}
for $x_1 \geq 0$ and $y \geq 0$. The exogenous disturbance disturbance input is assumed to be $\delta \in \mathcal{L}_2([0,\infty))$. To highlight fundamental tradeoffs due to autocatalytic structure of the system, we normalize the concentration such that {steady-states become 
\begin{equation}
{y}^*=1~~\text{and}~~~{x_1}^*=\frac{k_y}{ k}. 
\label{eq:1062}
\end{equation}}
As we discussed earlier, one may consider expression $\frac{2 y^{a}}{1+ y^{2h}}$ as the regulatory feedback control employed by nature that captures inhibition of the catalyzing enzyme. Hence, we can derive a control system model for glycolysis as follows
\begin{eqnarray}
\dot{x}_1 & = & -k \,x_1  +  u,  \label{cont-gly12}\\
\dot{y} & = &  (\alpha+1) k \, x_1 -  {\alpha}\, u -k_y \, y -\delta, \label{cont-gly22}
\end{eqnarray}
where $u$ is the control input.
Now, applying Theorem  \ref{theorem-5} to this model, it follows  that
\begin{equation}
\gamma ~>~ \frac{\alpha}{k+\alpha k_y}.
\label{trade}
\end{equation}
Equation (\ref{trade}) illustrates a tradeoff between robustness and efficiency (as measured by complexity and metabolic overhead).  From (\ref{trade}) the glycolysis mechanism is more robust efficient if $k$ and $k_y$ are large. On the other hand, large $k$ requires either a more efficient or a higher level of enzymes, and large $k_y$ requires a more complex allosterically controlled PK enzyme; both would increase the cell's metabolic load. 
We note that the existing hard limit is an increasing function of $\alpha$. This implies that increasing $\alpha$ (more energy investment for the same return) can result in worse performance. 
It is important to note that these results are consistent with results in \cite{BuziD10}, where a linearized model with a different performance measure is used.  
\end{example}

\section{Conclusion}\label{sec:conclusion}
The primary goal of this paper is to characterize fundamental limits on robustness and performance of a class of dynamical networks with autocatalytic structures. A simplified model of Glycolysis pathway is considered as the motivating application. We explicitly derive hard limits on the best achievable performance of the autocatalytic pathways with intermediate reactions which are characterized as $\mathcal L_2$-norm of the output as well as $\mathcal L_2$-gain of disturbance attenuation. 
Then, we explain how these resulting hard limits lead to some fundamental tradeoffs. For instance, due to the existence of autocatalysis in the system, a fundamental tradeoff between a notion of fragility (e.g., cell death)  and net production of the pathway emerges. Moreover, it is shown that as the number of intermediate reactions grows, the price paid for robustness increases. On the other hand, the larger the number of intermediate reactions involved in the breakdown of a metabolite, the less complex the enzymes involved in the individual reactions need to be. This illustrates a tradeoff between robustness and efficiency as measured by complexity and metabolic overhead.

\begin{spacing}{1}
\bibliographystyle{IEEEtran}
\bibliography{IEEEabrv,ref/TAC-bib-milad-feb2015} 
\end{spacing}

\end{document}